\documentclass[envcountsect]{llncs}
 \usepackage{listings}
 \usepackage{color}

 \usepackage{fullpage}
 \usepackage{amssymb}
 \usepackage{amsmath}
 \usepackage[utf8]{inputenc}
 \usepackage{graphicx}
 \usepackage{hyperref}
 \usepackage[blocks]{authblk}

 \spnewtheorem{thm}[theorem]{Theorem}{\bfseries}{\itshape}
 \spnewtheorem{lem}[theorem]{Lemma}{\bfseries}{\itshape}
 \spnewtheorem{prop}[theorem]{Proposition}{\bfseries}{\itshape}
 \spnewtheorem{cor}[theorem]{Corollary}{\bfseries}{\itshape}
 \spnewtheorem*{appx-lem}{Lemma}{\bfseries}{\itshape}
 \spnewtheorem*{appx-thm}{Theorem}{\bfseries}{\itshape}

 \spnewtheorem{defn}[theorem]{Definition}{\bfseries}{}

 \spnewtheorem*{rem}{Remark}{\itshape}{}
 \spnewtheorem{prob}[theorem]{Problem}{\itshape}{}
 \spnewtheorem{alg}[theorem]{Algorithm}{\itshape}{}

 \newcommand{\cost}{E}

\newcommand{\equalvolapprox}{(\frac{5}{2})^{\alpha -1}\tilde{B}_{\alpha}((1+\varepsilon)(1 + \frac{w_{\max}}{w_{\min}}))^\alpha}
\newcommand{\strictlyequalvol}{2(1+\varepsilon)(5(1+\varepsilon))^{\alpha-1}\tilde{B}_{\alpha}}
\newcommand{\lpgap}{12^{\alpha -1}}
\newcommand{\singleprocgap}{(12(1+\varepsilon))^{\alpha -1}}

\title{Energy-efficient algorithms for non-preemptive speed-scaling }

\author{
Vincent Cohen-Addad\inst{1} \and Zhentao Li\inst{1} \and
Claire Mathieu\inst{1}  \and Ioannis Milis\inst{2}}
\institute{D\'{e}partement d’Informatique, UMR CNRS 8548,
\'{E}cole Normale Sup\'{e}rieure, Paris, France\\
\email{\{vcohen, zhentao, cmathieu\}@di.ens.fr}
\and Department of Informatics, Athens University of Economics and Business, Greece\\
\email{milis@aueb.gr}
}

\begin{document}
\maketitle

\begin{abstract}
We improve complexity bounds for energy-efficient speed scheduling problems for both the single processor
and multi-processor cases. Energy conservation has become a major concern, so revisiting traditional
scheduling problems to take into account the energy consumption has been part of the agenda of the scheduling
community for the past few years~\cite{Albers_2010}.

We consider the energy minimizing speed scaling problem
introduced by Yao et al. \cite{YDS} where we wish to schedule a set of jobs,
each with a release date, deadline and work volume, on a set of identical processors.
The processors may change speed as a function of time and the energy
they consume is the $\alpha$th power of its speed. The objective is then to find a feasible schedule
which minimizes the total energy used.

We show that in the setting with an arbitrary number of processors where all work volumes
are equal, there is a $2(1+\varepsilon)(5(1+\varepsilon))^{\alpha -1}\tilde{B}_{\alpha}=O_{\alpha}(1)$ approximation
algorithm, where $\tilde{B}_{\alpha}$ is the generalized Bell number.
This is the first constant factor algorithm for this problem. This algorithm extends to general unequal processor-dependent work volumes, 
up to losing a factor of $(\frac{(1+r)r}{2})^{\alpha}$ in the approximation, where $r$ is the maximum ratio between two work volumes. 
We then show this latter problem is APX-hard, even in the special case when all release dates and deadlines are equal and $r$ is 4.

In the single processor case, we introduce a new linear programming formulation
of speed scaling and prove that its integrality gap  is at most $\lpgap$. As a corollary,
we obtain a $\singleprocgap$ approximation algorithm where there is a single processor,
improving on the previous best bound of $2^{\alpha-1}(1+\varepsilon)^{\alpha}\tilde{B}_{\alpha}$
when $\alpha \ge 25$.


\end{abstract}

\section{Introduction}
While traditional scheduling problems aim to process jobs as quickly as possible given a variety of side constraints, energy-efficient scheduling aims to also minimize the energy consumed by the system, typically by changing processor's frequency to scale its speed dynamically, slowing it  down at times to conserve energy. Thus, standard scheduling problems must now be revisited to take energy into account, and this has been part of the agenda of the scheduling community for the past few years (see the survey~\cite{Albers_2010} and the references therein).

In minimum energy scheduling problems, introduced by Yao et al. \cite{YDS}, we wish to execute jobs on a single (or a set of) processor(s) so that all jobs complete between their release date and deadline in such a way that minimizes the energy consumed. Now, each job has to execute a work volume $w$ and as the processors may change their speed, a job may be completed faster (or slower) than the time $w$ it needs to execute at speed 1. It is observed that a processor running at speed $s$  consumes power  at the rate $s^\alpha$, for a constant $\alpha>1$ (typical values of  $\alpha$ are less than 3) and so a processor running at speeds $s(t)$ during an interval $I$ would consume energy  $\int_{t\in I} s(t)^\alpha dt$. 

\subsubsection{Problem definition.}

In this paper we examine a minimum  energy scheduling problem, which in its simplest \emph{non preemptive} form can be stated as follows.

\begin{prob}\label{prob:nonpreemptive}
  \textsc{Non preemptive minimum energy scheduling ($\alpha$)}
  \\
  \textbf{Input:}\\
   - $m$ processors $P=\{p_1,p_2,\ldots ,p_m\}$. \\
   -  $n$ jobs $J=\{1,\ldots,n\}$, with a life interval $L_j=[r_j, d_j]$ and a work volume $w_j$ for each $j \in J$. \\
  \textbf{Output:}
  An assignment $S$ of an \emph{execution interval} $[s_j,e_j] \subseteq [r_j, d_j]$ for each job $j$ such that no $m+1$ execution intervals have a common intersection.
  ~\\
  \textbf{Objective:} Minimize 
$ E(S) = \displaystyle  \sum_j (e_j-s_j)\left( \frac{w_j}{e_j-s_j}\right) ^\alpha $
\end{prob}

An assignment $S$ of execution intervals for all jobs is called a \emph{schedule}. Equivalently, we could ask an algorithm to also output an assignment of each job to a processor in $P$ but this assignment is obtained greedily from the output above. By convexity of $s \rightarrow s^\alpha$, it is more efficient to run a processor at constant speed for the same job. Hence, the energy consumed by a job $j$ with execution interval $I$ is $E(S,j)=|I|\cdot (w_j/|I|)^\alpha$, and we can think of ${w_j}/{|I|}$ as the speed given to job  $j$. Clearly, $E(S)=\sum_j E(S,j)$.

For a job $j$ with life interval $[r_j, d_j]$, we say $r_j$ is the \emph{release date} of $j$ and $d_j$ is the \emph{deadline} of $j$. Hence $j$ must be executed between $r_j$ and $d_j$. 
We say that $j$ is \emph{alive} at time $t$ if $t\in L_j$. In the {\em preemptive} case, we are allowed to stop a job, execute some other job and restart the first job later on. Equivalently, we can think of breaking each job $j$ into as many pieces as we want, that all have the same life interval as $j$, their total work volume is the work volume of $j$ and they are executed non-preemptively.  In the \emph{migratory} version of the problem, stopped jobs can even continue their execution on a different processor.

\subsubsection{Related work.} 

The single processor preemptive  problem is polynomial: Yao et al. \cite{YDS}   proposed an elegant greedy algorithm whose optimality was proved  by Bansal et al. \cite{BKP}. On the other hand, the  single processor non-preemptive problem is NP-hard (Antoniadis and Huang \cite{Antoniadis_Huang}), even for instances where for any pair of  jobs such that $r_j \leq  r_{j'}$, it holds that $d_j \geq  d_{j'}$; they also proposed a $(2^{5\alpha-4})$-approximation algorithm for general instances. 
Moreover,  a $(1 + \frac{w_{\max}}{w_{\min}})^{\alpha}$-approximation algorithm for general  instances of this problem proposed in  \cite{BKLLN}, while \cite{BKLLS} introduced a $ 2^{\alpha-1}(1+\varepsilon)^\alpha \tilde{B}_\alpha$-approximation which is better for small values of $\alpha$, where $\tilde{B}_{\alpha}$ is the generalized version of the Bell number 
introduced by \cite{BKLLS} which is equal to $\tilde{B}_{\alpha} = \sum\limits_{k=0}^{\infty} \frac{k^{{\alpha} -1} e^{-1}}{k!} < \left( \frac{e^{-0.6 + \varepsilon} {\alpha}}{\ln({\alpha}+1)}\right)^{\alpha}$,
and $\tilde{B}_{\alpha} = \Omega \left( ( \frac{\alpha}{e \text{ln} \alpha} ) ^\alpha \right)$, for any $\alpha \in \mathbb{R}^{+}$ \cite{bell_numb}.

The homogeneous multiprocessor preemptive problem remains polynomial  
when  migration of jobs is allowed \cite{Albers_et_al_SPAA11,Bampis_et_al_EUROPAR12}. 
However, \cite{Albers_et_al_SPAA07} proved that the non-migratory variant of this problem, even for jobs with common release dates and deadlines, is NP-Hard and gave a PTAS for such instances.   
Greiner et al. \cite{Greiner_et_alii} proposed a transformation of an optimal solution 
to general migratory instances to a $B_{\lceil \alpha \rceil}$-approximate
solution for non-migratory problem. For the  homogeneous multiprocessor non-preemptive problem  Bampis et al. \cite{BKLLN}  proposed a  $m^{\alpha -1}(n^{1/m})^{\alpha -1}$-approximation algorithm. 

Bampis et al. \cite{BKLLS} studied the  heterogeneous multiprocessor preemptive problem where 
every processor $i$ has a different speed-to power function, $s^{\alpha(p_i)}$, and 
both the life interval and the work of jobs are processor dependent. For the migratory variant they proposed a polynomial in $\frac{1}{\varepsilon}$ algorithm returning a solution within an additive factor of $\varepsilon$  far from the optimal solution, and for non-migratory variant  an  ($((1+\frac{\varepsilon}{1 - \varepsilon})(1 + \frac{2}{n-2}))^\alpha \tilde{B}_\alpha$)-approximation algorithm.

 \begin{table}[htb]
  \begin{center}
    \begin{tabular}{|l|l|l|}
      \hline
      Problem & Complexity & Approximation Ratio\\
      \hline
       \hline
      $1|r_j,d_j, \text{pmtn}|E$ & Polynomial \cite{YDS} & \\ \hline
      $1|r_j, d_j|E$ & NP-Hard \cite{Antoniadis_Huang}& $2^{5\alpha-4}$ \cite{Antoniadis_Huang} \\
      & &  $(1 + \frac{w_{\max}}{w_{\min}})^{\alpha}$ \ \cite{BKLLN}\\
      & &  $          2^{\alpha -1}   (1+{\varepsilon})^{\alpha} \tilde{B}_\alpha  $\ \cite{BKLLS}\\
& &  ${(12(1+\varepsilon))^{\alpha -1}}$ {\bf [this paper]} \\
      \hline
      \hline
      $P|r_j,d_j, \text{pmtn}, \text{mig}|E$ & Polynomial \cite{Albers_et_al_SPAA11,Bampis_et_al_EUROPAR12}& \\ \hline
      $P|r_j = 0, d_j = 1, \text{pmtn}, \text{no-mig}|E$ & NP-Hard \cite{Albers_et_al_SPAA07}& PTAS \cite{Albers_et_al_SPAA07}\\ \hline
      $P|r_j,d_j, \text{pmtn}, \text{no-mig}|E$ & NP-Hard &  $B_{\lceil \alpha \rceil}$ \cite{Greiner_et_alii}  \\ \hline
      $P|r_j,d_j|E$ & NP-Hard & $m^{\alpha} (n^{1/m})^{\alpha - 1}$ \cite{BKLLN}\\
      & & ${ \equalvolapprox}$ {\bf  [this paper]}\\ \hline
 $P|r_j=0,d_j=1, w_{i,j},\text{pmtn}, \text{no-mig}|E$ & APX-hard {\bf [this paper]} & \\
      \hline
	\hline
    $R|r_{ij},d_{ij}, w_{ij},\text{pmtn}, \text{mig}|E$ & Polynomial$(\frac{1}{\varepsilon})$ \cite{BKLLS} & $OPT+\varepsilon$ \cite{BKLLS}\\
    $R|r_{ij},d_{ij}, w_{ij},\text{pmtn}, \text{no-mig}|E$ & NP-Hard & $(1+{\varepsilon})^{\alpha} \tilde{B}_\alpha$ \cite{BKLLS}\\

\hline
    \end{tabular}
  \end{center}
\label{related}
  \caption{Known and our (in bold) results for minimum energy scheduling problems. Problems are denoted by extending  the standard three-field notation of Graham et al. \cite{Graham_et_al}.  $P$ denotes a  homogeneous multiprocessor where all  processors obey the same speed-to-power function $s^{\alpha}$, while $R$ is used to denote a heterogeneous multiprocessor where each processor has its own 
speed-to-power function $s^{\alpha(p_i)}$. For both environments  the work volume of each job may depend on the processor it is executed and this is indicated by including $w_{ij}$ in the second field. }
\end{table}

In Table \ref{related} we summarize the results mentioned above and our contribution (in bold).
There are also results for special cases of the energy minimization problems when jobs have life intervals of a specific structure (common, agreeable, laminar, purely laminar) or/and equal work volumes ~\cite{YDS,Albers_et_al_SPAA07,BKLLN,Antoniadis_Huang,Huang_Ott}. 
Some of the works mentioned above ~\cite{YDS,BKP,Albers_et_al_SPAA11,Albers_et_al_SPAA07} as well as ~\cite{Bansal_et_al_Algorithmica11,Bansal_et_al_TOCS12} study  online algorithms for energy minimization problems in the speed scaling setting on a single processor or homogeneous multiprocessors.


\subsubsection{Our results.}

In Section~\ref{sec:equal-weights}, we give a $\equalvolapprox$-approximation to the \textsc{Non preemptive minimum energy scheduling ($\alpha$)} problem where $w_{\max}=\max_i w_i$ and $w_{\min}=\min_i w_i$ (Theorem~\ref{thm:equal-weights}). This is the first multiprocessor algorithm with an approximation factor  independent of $n$ and $m$, improving on the previous approximation of $m^{\alpha -1}(n^{1/m})^{\alpha -1}$~\cite{BKLLN}. Our ratio becomes $2(1+\varepsilon)(5(1+\varepsilon))^{\alpha-1}\tilde{B}_\alpha$ when all jobs have the same work volume and this is also the first constant approximation factor for this case.  Recall that Albers et al. \cite{Albers_et_al_SPAA07} showed the preemptive non-migratory  version of  
this problem is NP-Hard and Greiner et al. \cite{Greiner_et_alii} gave a $B_{\lceil \alpha \rceil}$-approximation algorithm for it. However the non-preemptive version resisted so far.

Up to an additional factor of $(w_{\max}/w_{\min})^\alpha$, this extends to the case where the work of jobs $w_{ij}$ depends on the processor $i$ on which $j$ is executed. In Section~\ref{section:apx}, we prove (Theorem~\ref{thm:apx_hard})  that this latter problem is APX-hard even for jobs with common life intervals and work volume in $\{1,3,4\}$.  This is the first APX-hardness result for an energy minimization problem.

In Section~\ref{sec:single-proc}, we prove (Theorem~\ref{thm:intgap}) that a natural LP relaxation for  \textsc{Non preemptive minimum energy scheduling ($\alpha$)}  on a single processor has integrality gap at most $\lpgap$. Our LP relaxation is obtained from the \emph{compact} LP relaxation in the preemptive setting of~\cite{BKLLS} (equivalent to their configuration LP)  by adding a constraint capturing non-preemption (otherwise the integrality gap is unbounded, see Lemma~\ref{lem:gap_lp1}).
Our result is the first LP relaxation with a gap independent of  $n$ and the work $w_j$ of the jobs. 
As a corollary, we obtain a $(12(1+\varepsilon))^{\alpha -1}$
approximation to the \textsc{Non preemptive minimum energy scheduling ($\alpha$)}
problem on a single processor.  
Compared to the previous best constant factor approximations of $\min\{2^{\alpha-1 }\tilde{B}_\alpha , 2^{5\alpha -4}\}$\cite{BKLLS,Antoniadis_Huang}, this is always better then $2^{5\alpha -4}$ and better than $2^{\alpha-1 }\tilde{B}_\alpha$ for any $\alpha\geq 25$.

\subsubsection{Preliminaries}

\begin{defn} An {\em independent} set of jobs is a set of jobs whose life intervals do not mutually intersect.
\end{defn}

\begin{defn} We say that an independent set is \emph{good} if the life interval of no job falls between the deadlines of two consecutive jobs of this independent set.
\end{defn}

\begin{prop}\label{lemma:half} \cite{Antoniadis_Huang}
    Let $S$ and $S'$ be two schedules that schedule job $j$ with an execution interval $I$ and $I'$ respectively. Then $E(S', j) = (|I|/|I'|)^{\alpha-1} E(S, j)$.

\end{prop}


\section{Multiprocessor scheduling}\label{sec:equal-weights}
In this section we present an approximation algorithm for \textsc{Non-preemptive minimum energy scheduling ($\alpha$)}  and show the following theorem:
\begin{thm}\label{thm:equal-weights}
{There exists a polynomial-time approximation algorithm for  \textsc{Non preemptive minimum energy scheduling ($\alpha$),}}
with approximation factor
    $(\equalvolapprox)$.
\label{EVJ}
\end{thm}
In particular, when all jobs have the same work volume, this yields a schedule which consumes energy  within 
$\strictlyequalvol$
of the optimal energy, a factor  depending only on $\alpha$ and independent 
from the number $n$ of job and of the number $m$ of processors. 

Our algorithm uses a reduction to  the following problem, previously studied by  Bampis et al. \cite{BKLLS}.
\begin{prob}
  \textsc{Non-preemptive fully heterogeneous minimum energy scheduling $(\alpha)$}\\
  \textbf{Input:}\\
  - $m$ heterogeneous processors $P=\{p_1,p_2,\ldots,p_m\}$, and a number $\alpha(p_i) \le \alpha$ for each one.\\
  -    $n$ jobs $J=\{1,\ldots,n\}$, a {life interval} $L_{ij}=[r_{ij}, d_{ij}]$ and a  work volume $w_{ij}$ for each $j \in J$, $p_i \in P$.\\
  \textbf{Output:}
  An assigment $S$ of a  processor $p(j)$ and an execution interval 
$[s_j,e_j] \subseteq [r_{ij}, d_{ij}]$  for each job $j \in J$ such that for each pair of jobs $j_1,j_2$, if $p(j_1) = p(j_2)$ then   $[s_{j_1}, e_{j_1}] \cap [s_{j_2}, e_{j_2}] = \emptyset$.\\
  \textbf{Objective:} Minimize
$\displaystyle  E(S)=\sum_{p_i} \sum_{j:p(j)=p_i} (e_j-s_j)\left( \frac{ w_j}{e_j-s_j}\right) ^{\alpha(p(j))}$.
\label{FHproblem}
\end{prob}

\begin{thm}\cite{BKLLS}\label{thm:previous}
There is an approximation algorithm for the  \textsc{Non-preemptive fully heterogeneous minimum energy scheduling} problem without migration with approximation ratio
$(1+\varepsilon)^{\alpha}\tilde{B}_{\alpha}$.
\end{thm}

\subsection{Overview}

The algorithm proceeds as follows: we consider the life intervals of all the jobs, greedily find $m$ maximal independent sets, 
 and assign the jobs of the $i$th independent set $\mathcal{J}_i$ to processor $p_i$. Then we partition time on  $p_i$ according to the deadlines of the jobs in $\mathcal{J}_i$, and restrict ourselves to schedules such that no execution interval on $p_i$ overlaps such a deadline. We solve the resulting restricted problem using the algorithm from Bampis et al.'s~\cite{BKLLS} to obtain a feasible schedule.

To analyze its cost, we first show that an optimal solution can be transformed into a solution satisfying our additional constraints and without increasing the cost by too much. 
We start with an optimal solution and attempt, for each job $j$ in the $i$-th independent set,  to  move $j$ to processor $p_i$ and execute it in the middle fifth of its life interval. Its execution interval is then shrunk by a factor of at most 5 and, by Proposition~\ref{lemma:half}, its energy consumption is increased by a factor of at most $5^{\alpha-1}$. If we are unable to do so for some $j$, it is because of some other job $j'$ on processor $p_i$ with a significant overlap with $j$, and we execute both $j$ and $j'$ during the time of overlap.




To guarantee that no execution interval on $p_i$ crosses one of our selected deadlines, we argue that each execution interval crosses at most one such deadline and further modify the schedule, restricting the execution interval to one of the two sides of the deadline, up to shrinking its execution interval by a factor of 2.


Finally, the algorithm of Bampis et al. \cite{BKLLS} provides an approximation to our constrained problem.

\subsection{Scheduling algorithm}\label{subsection:algorithm}

We now give a detailed description of our algorithm.

\begin{alg}\label{alg:equalweight}
~
\begin{enumerate}
\item
$R\gets J$
\item
For $i=1$ to $m$:
    \begin{enumerate}
    \item $\mathcal{I}_i\gets \emptyset$, $k\gets 0$ and $t_0^i\gets -\infty$.
    \item While  $\exists j\in R$  such that $\{ j\} \cup \mathcal{I}_i$ is an independent set
    \item ~~~ Find such a $j$ with $d_j$ minimum and let $\mathcal{I}_i \gets \mathcal{I}_i\cup \{ j\}$ and $R\gets R\setminus \{ j\}$.
    \item ~~~ $k\gets k+1$ and $t_k^i\gets d_j$. 
    \item $t_{k+1}^i\gets +\infty$
    \end{enumerate}
\item
For every processor $p_i$, for $k=1$ to $|\mathcal{I}_i |+1$, let $I_k^i=[  t_{k-1}^i , t_{k}^i ]$.
\item
Create an instance of Problem \ref{FHproblem} as follows:
    \begin{enumerate}
    \item
    For every processor $p_i$, for every interval $I_{l}^i$, create a heterogeneous processor $(i,l)$ with $\alpha_{i,l} = \alpha$.
    \item
    For every job $j \in J$ which is alive during  part or all of some $I_{l}^i$, set\\
$~~~$  release date   $  r_{(i,l)j} = \max ( r_j - t_{l-1}^i  ,0)   $, deadline    $   d_{(i,l)j} = \min ( d_j-  t_{l-1}^i ,  t_{l}^i - t_{l-1}^i  )    $, work $w_{(i,l)j} = w_j$
    \end{enumerate}
\item
Solve the created problem using the algorithm from~\cite{BKLLS}\\
Let $J_{i,l}$ be the set of jobs scheduled (preemptively) on heterogeneous processor $(i,l)$.
\item
For each $(i,l)$, reorder the execution intervals inside $I^i_l$ \\
$~~$ so that the jobs of $J_{i,l}$ are executed by order of non-decreasing deadline.
\end{enumerate}
\end{alg}

\subsection{Analysis}
We first prove the following lemma which has a crucial role in the analysis of the approximation ratio.

\begin{lem}\label{lemma:independent}
Let $\{\mathcal{J}_1,...,\mathcal{J}_m\}$ be a subpartition of $J$ such that each $\mathcal{J}_i$ is an independent set of jobs, and $S$ be a schedule of $J$. Then there exists a schedule $S'$  such that for every $i$ all the jobs of $\mathcal{J}_i$ are executed on processor $p_i$,
and whose cost satisfies  $\cost (S')\leq (5/2)^{\alpha -1} (1 + \frac{w_{\max}}{w_{\min}})^{\alpha} \cost (S)$   $\forall i \in \{1,\ldots,m\}$.
\end{lem}
\begin{proof}
  Let $\mathcal{J}=\mathcal{J}_1\cup \ldots \cup \mathcal{J}_m$, and, for $j\in \mathcal{J}_i$, let  $I_j$ denote the execution interval of job $j$ in schedule $S$. Assume that in $S$ job $j$ is executed on a processor other than $p_i$.
  We distinguish two cases:
  \begin{enumerate}
  \item
    If there exists a job $j'$ executed on $p_i$ and such that  $|I_{j'} \cap I_j|\geq \frac{2}{5}\min ( |I_j|, |I_{j'}|)$,
    then in $S'$, we schedule both  jobs $j$ and $j'$ on $p_i$ during the interval $I_{j'} \cap I_j$, so that the energy used is minimized. Processor $p_i$ runs at a constant speed during $I_{j'} \cap I_j$, executing total work $w_j + w_{j'}$, and the energy consumed is therefore
    $\frac{(w_j + w_{j'})^{\alpha}}{(|I_{j'} \cap I_j|)^{\alpha -1}}$.

    For the cost analysis, by symmetry we may assume that $|I_j| \le |I_{j'}|$, and note that $E(S,j)=(w_j)^{\alpha}/(|I_j|)^{\alpha-1}$. We then write:
$$\frac{(w_j + w_{j'})^{\alpha}}{|I_{j'} \cap I_j|^{\alpha -1}}     
    \leq  \frac{(w_j + w_{j'})^{\alpha}}{(2|I_j|/5)^{\alpha -1}}
    =   \left(\frac{5}{2}\right)^{\alpha -1} \left(1 + \frac{w_{j'}}{w_j}\right)^\alpha \cdot \frac{w_j^{\alpha}}{|I_j|^{\alpha -1}}
    \leq   \left(\frac{5}{2}\right)^{\alpha -1} \left(1 + \frac{w_{\max}}{w_{\min}}\right)^\alpha \cdot E(S,j).    
    $$

\item
      Otherwise, on processor $p_i$ no processor $j'$ has $I_{j'}\subseteq I_j$, so on $p_i$ during $I_j$ the execution has at most two jobs $j'$, one whose execution interval overlaps the start and the other whose execution interval overlaps the end, and there must be an idle interval of size at least $\frac{1}{5}|I_j|$ in the center of $I_j$: then in $S'$ we schedule job $j$ on $p_i$ during this interval. 
      
      For the cost analysis, the execution interval of $j$ is shrunk by a factor of at most $5$ so
    $$E(S',j)= 5^{\alpha -1} E(S,j) \le \left(\frac{5}{2}\right)^{\alpha -1} \left(1 + \frac{w_{\max}}{w_{\min}}\right)^\alpha E(S,j).$$
  \end{enumerate}
 \qed\end{proof}



We now show that we can force every job that is executed on processor $p_i$  to be scheduled during a subinterval of some $[  t_{\ell-1}^i , t_{\ell}^i ]$. 

\begin{lem}\label{lemma:cut}
 Let $\{\mathcal{J}_1,...,\mathcal{J}_m\}$  and $(I_l^i)_{i,l}$
 be defined as in Algorithm~\ref{alg:equalweight}. Let $S'$ be a schedule
 such that for every $i$ all the jobs  of  $\mathcal{J}_i$ are executed on processor $p_i$ .
 There exists a schedule $S''$ such that for each processor $p_i$ and for each job $j$ that is executed on $p_i$ in $S'$,
in $S''$ $j$ is executed on $p_i$ and  the execution interval of $j$ is included in some $ I_l^i$;
 moreover, for any $j$, $E(S'',j)\leq 2^{\alpha -1} E(S', j)$. 
\end{lem}
\begin{proof}
  We consider the schedule $S'$.
  We first show that the execution interval of a job $j \in R$ that is executed on processor $i$ cannot intersect more than
  two intervals $I_l^i$. Assume, for a contradiction, that there is a job $j$ on processor $i$ whose execution interval intersects at least three intervals
  $I_{l-1}^i$, $I_{l}^i$ and $I_{l+1}^i$. Then it must contain $I_l^i$ entirely, contradicting the fact that there exists a
  job of $\mathcal{J}_i$ that must be executed on processor $i$ during $I_l^i$.

  Therefore, $[s_j,e_j]\subseteq I_{l}^i \cup I_{l+1}^i$. Write $[s_j,e_j]=[s_j,t^i_{l}]\cup [t^i_{l},e_j]$. Cut its execution
  interval at $t^i_l$ and keep the larger of the two subintervals. By Proposition~\ref{lemma:half}, this leads to a schedule $S''$ which satisfies the desired property.
\qed\end{proof}
\begin{figure}[ht!]
  \begin{center}
    \includegraphics[scale=0.35]{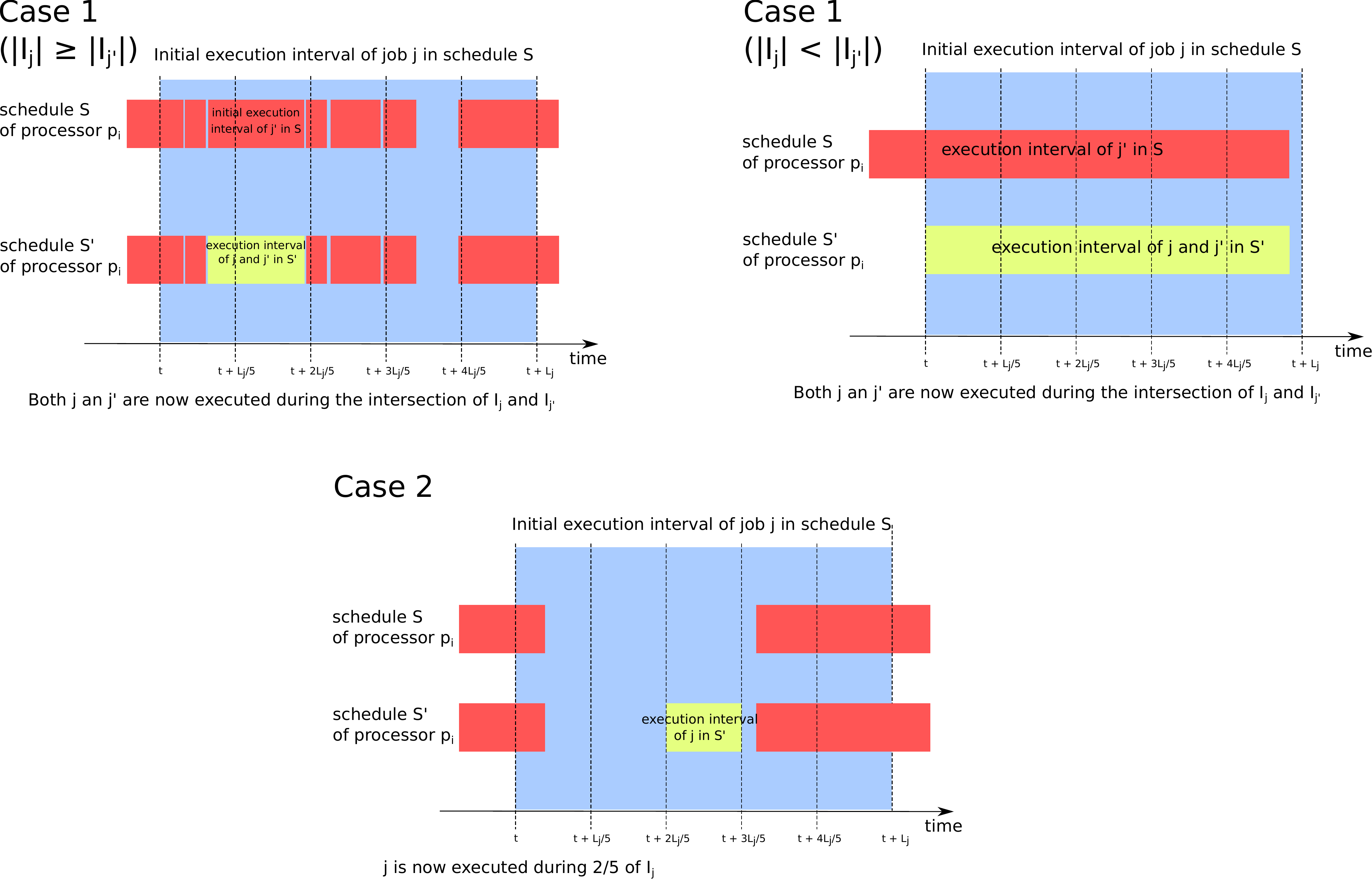}
  \end{center}
  \caption{The two different cases for lemma \ref{lemma:independent}. $S$ is the original schedule and $S'$ the schedule we build during the proof
    of the lemma.}
  \label{fig:independent}
\end{figure}

Looking at Lemma~\ref{lemma:independent}, we  notice that the jobs whose cost
  has changed between $S$ and $S'$ are, in $S'$, now executed on $p_i$ during the life interval of an element of $\mathcal{J}_i$. Looking at Lemma~\ref{lemma:cut}, we notice that the jobs whose cost
  has changed between $S'$ and $S''$ are, in $S'$, executed on $p_i$ in an interval that overlaps one of the $t_{\ell}^i$. Thus for every  $j$ we  have   $E(S,j)=E(S',j)$ or $E(S',j)=E(S'',j)$. Hence the cost of a job
after the two tranformations from $S$ to $S'$ to $S''$  increases by a factor of at most:
$$\frac{\cost (S'')}{\cost (S)}\leq \max\left(\left(\frac{5}{2}\right)^{\alpha -1} \cdot \left(1 + \frac{w_{\max}}{w_{\min}}\right)^{\alpha}, 2^{\alpha-1} \right) =
\left(\frac{5}{2}\right)^{\alpha -1} \cdot \left(1 + \frac{w_{\max}}{w_{\min}}\right)^{\alpha}.$$
Finally, observe that the last step of Algorithm~\ref{alg:equalweight} transforms the schedule into a non-preemptive schedule that is feasible and has the same cost. Putting Lemma~\ref{lemma:independent}, Lemma~\ref{lemma:cut} and Theorem~\ref{thm:previous} together, we obtain an approximation ratio of
$$\left(\frac{5}{2}\right)^{\alpha -1} \cdot \left((1+\varepsilon)\left(1 + \frac{w_{\max}}{w_{\min}}\right)\right)^{\alpha} \cdot \tilde{B}_{\alpha}.$$
This completes the proof of Theorem \ref{EVJ}.

\begin{cor}{There is a polynomial time algorithm which computes a $2(1+\varepsilon)(5(1+\varepsilon))^{\alpha-1}\tilde{B}_\alpha$-approximation
to the  \textsc{Non preemptive minimum energy scheduling}($\alpha$) for jobs of equal work volume.}
\end{cor}

\paragraph{Remark.} By an easy reduction, Theorem \ref{EVJ}  extends to the case where the work volume of each job, $w_{ij}$, depends on the processor $p_i$ on which job $j$ is executed, up to losing an additional factor of $(w_{\max} / w_{\min})^\alpha $ in the approximation ratio.

\section{Hardness of approximation}
\label{section:apx}

In this section, we show the following problem mentioned in the
previous remark is APX-hard by a reduction from \textsc{Maximum
Bounded 3-Dimensional Matching}.

\begin{prob}
  \textsc{Minimum energy scheduling with processor dependent works} $(\alpha)$\\
  \textbf{Input:}\\
  - $m$ processors $P=\{p_1,p_2,\ldots ,p_m\}$.\\
  - $n$ jobs $J=\{1,\ldots,n\}$,  a life interval $L_j=[r_j, d_j]$ and a work volume $w_{ij}$ for each $j \in J,~p_i \in P$.\\
\textbf{Output:}
  An assignemnt $S$  of a processor $p(j)$ and an execution interval $[s_j,e_j] \subseteq [r_j, d_j]$ for each job $j$   such that for each pair of jobs $j_1,j_2$, if $p(j_1) = p(j_2)$ then
  $[s_{j_1}, e_{j_1}] \cap [s_{j_2}, e_{j_2}] = \emptyset$. \\
  \textbf{Objective:} Minimize
$\displaystyle E(S)=  \sum_{p_i} \sum_{j:p(j)=p_i} (e_j-s_j)\left( \frac{ w_{ij}}{e_j-s_j}\right) ^{\alpha}$.
\label{APXproblem}
\end{prob}

\begin{thm}\label{thm:apx_hard}
  Minimum energy scheduling with processor dependent work is APX-Hard.
\end{thm}

\begin{prob}
  \textsc{Maximum Bounded 3-Dimensional Matching }
  \\
  \textbf{Input:} Three sets $A,B,C$ of equal cardinality, and a subset ${\cal T}$ of $A \times B \times C$, such that every element in $A\cup B\cup C$ appears in at least one and at most  three elements of ${\cal T}$.

\noindent
  \textbf{Output:}
A subset ${\cal T}'$ of ${\cal T}$, such that each element of $A\cup B\cup C$ occurs in at most one triple of ${\cal T}'$.
  ~\\
  \textbf{Objective:} Maximize the cardinality of ${\cal T}'$.
\end{prob}

\begin{thm}\label{thm:petrank}
 \cite{Petrank} Max-Bounded-3DM is APX-hard, even when restricted to
instances where the optimal solution has size $q = |A| = |B| = |C|$.
\end{thm}



\begin{proof}[of Theorem \ref{thm:apx_hard}]
We describe a reduction from the Maximum Bounded 3-Dimensional Matching problem (Max-Bounded-3DM) with optimal solution of size $q$,
to the  \textsc{Minimum energy scheduling with processor dependent works} $(\alpha)$ problem.
Our construction draws some ideas from Azar et al. \cite{Azar_Epstein_Richter_all-norm} and
Lenstra et al. \cite{Lenstra_Shmoys_Tardos_approx_unrelated}.

We first define a function $f$ that maps each instance $I$ of Max-Bounded-3DM to
an instance $I'$ of  \textsc{Minimum energy scheduling with processor dependent works} $(\alpha)$ problem,
and a function $g$ that maps a feasible schedule for $I'$ to a feasible solution for $I$.
Both are computable in polynomial time.

The function $f$ creates a scheduling instance $I'$ as follows.
\begin{itemize}
\item  There are $3q$ machines:  $|{\cal T}|$ machines, one for each triple $T$,
 called {\em triple} machines,  and $3q-|{\cal T}|$ identical machines  called {\em dummy} machines.
\item There are  $5q$ jobs:  $3q$ jobs,  one job $j(e)$ for each element
 $e$ of $A\cup B\cup C$  called {\em element} jobs. Each element job $j(e)$ has work  $1$ on every triple machine whose corresponding $T$ contains $e$,  and work  $4$ on all other machines. The remaining $2q$ jobs are called  \emph{dummy} jobs, and they all have work $3$ on
all machines. All jobs have release date 0 and deadline $3$.
\end{itemize}

Let ${\cal T}^*$ denote an optimal set of triples for instance $I$. Recall that $|T^*|=q$ and consider the following schedule for  $I'$:
For each triple $T \in {\cal T}^*$
we schedule the three element jobs $j(e)$ with $e\in T$ on the triple machine for $T$.
Then, we schedule each of the $2q$ dummy jobs on one of the remaining $3q-|{\cal T}^*|=2q$ machines.
In the resulting schedule, every machine has to execute a total work load  of $3$  within $3$ units of time, so the energy of this schedule is $3\cdot 3q = 9q$. Therefore,
\begin{equation}\label{eq:boundOPT}
\text{OPT}(I') \le 9 \cdot \text{OPT}(I).
\end{equation}

To define $g$, given a feasible schedule $S'$ for an instance $I'=f(I)$,
we first apply the algorithm in Lemma \ref{lem:elementjobs} to $S'$ which constructs a new schedule $S$ with $E(S) \le E(S')$ and where each element job $j(e)$ is processed on a triple machine for $T$ with $e\in T$.
Then, we construct a 3-Dimensional Matching for $I$ by picking all triples $T$ whose three elements are all processed
on the machine for $T$ in $S$, so $|g(S')|$ is the number of triple machines that process all three of the corresponding jobs in $S$. 

We now show that there exists a positive constant $\beta$ such that for every schedule $S'$ for $I'$,
$ \text{OPT}(I) - |g(S')|  \le \beta ( E(S') - \text{OPT} (I') )$.
Since $E(S) \le E(S')$ and  $g(S) = g(S')$, 
it suffices to show that $ \text{OPT}(I) - |g(S)|  \le \beta ( E(S)- \text{OPT} (I') )$.


For $k = 0, 1, 2, 3$ let $m_k$ denote the number of machines in schedule $S$ that process exactly $k$ element jobs.  We have: $ \text{OPT}(I) - |g(S)| =q-m_3$.  In $S$ the total work load of all machines together is exactly $9q$, since each element job has work $1$. Only the machines with $0$ or $3$ element jobs may have work load equal to $3$ in $S$. All other machines have work load either less than or equal to  2 or greater than or equal to 4. Let $\ell_1,\ell_2,\ldots ,\ell_{3q}$ be the vector of the  loads on all the machines. A relaxation of the energy minimization problem is:
$$\min \sum_i \ell_i^\alpha / 3^{\alpha -1}  \hbox{ s.t. }\left\{ \begin{array}{ll}
\sum_i \ell_i &=9q\\
\ell_1=\cdots =\ell_{m_0+m_3} &= 3\\
\forall i>m_0+m_3, ~~\ell_i\leq 2 &\hbox{ or }\ell_i\geq 4
\end{array} \right. $$
By convexity of $\ell\mapsto \ell^\alpha$, the optimal solution  of the relaxation cannot have one load strictly less than 2 and another strictly more than 4, and neither can it have two different loads that are both $\geq 4$ or both $\leq 2$, so the only possibilities are all machines to have loads in $\{ 3,2 ,4\}$,  $\{ 3,1,4\}$ and  $\{3,2,x\}$ with $x>4$. By convexity of $\ell\mapsto \ell^\alpha$, one can check that the first possibility minimizes the energy, so we have
$$E(S)\geq [(m_0+m_3)3^{\alpha}+(3q-m_0-m_3)( (1/2)2^{\alpha}+(1/2)4^\alpha)]/3^{\alpha -1},$$
and so
$$E(S)- \text{OPT} (I') \geq (3q-m_0-m_3)( (1/2)2^{\alpha}+(1/2)4^\alpha - 3^\alpha)/3^{\alpha -1}.$$
Since the total number of element jobs as well as the total number of machines is $3q$, we have
$$m_1 + 2 \cdot m_2 + 3 \cdot m_3  = 3q =m_0+m_1+m_2+m_3,$$
and we deduce  $3q-m_0-m_3=m_1+m_2\geq (3/2)(q-m_3)$, hence
$$E(S)- \text{OPT} (I') \geq (3/2) \frac{ (1/2)2^{\alpha}+(1/2)4^\alpha - 3^\alpha)}{3^{\alpha -1}} ( \text{OPT}(I) - |g(S)|).$$
Noticing that $x\mapsto x^{\alpha}$ is strictly convex for $\alpha >1$, we obtain the desired inequality.
\qed \end{proof}

\begin{lem}\label{lem:elementjobs}
There exists a polynomial time algorithm which takes a schedule $S'$ for an instance  $I'=f(I)$ (for an instance $I$ of Max-3DM) and builds a feasible schedule $S$ whose cost satisfies $E(S) \le E(S')$, and where each element job $j(e)$ is scheduled on a triple machine
  whose triple $T$ contains $e$.
\end{lem}

\begin{proof}
  The algorithm proceeds as follows:\\
  \begin{tabular}{l}
    $S\gets S'$.\\
    For each element job $j(e)$,\\
    \quad If $j(e)$ is scheduled on a machine $p_i$ that is not a triple machine whose triple contains $e$, then \\
    \quad \quad Find any triple machine $p_{i_e}$ whose triple $T$ contains $e$.\\
    \quad \quad If $S$ schedules a job $j$ on $p_{i_e}$ that is either a dummy job or an element job of work volume 4, then\\
    \quad \quad \quad Swap $j$ and $j(e)$ in $S$ (and rebalance $p_i$ and $p_{i_e}$ so they run at constant speed).\\
    \quad \quad Otherwise,\\
    \quad \quad \quad Schedule $j(e)$ on $p_{i_e}$ instead of $p_i$ in $S$ (and rebalance as needed.)\\
  \end{tabular}

  \smallskip

  Since element jobs are never moved more than once, the algorithm produces a feasible schedule $S$ where all element jobs $j(e)$ are on a triple machine whose triple $T$ contains $e$. It remains to show that $E(S) \le E(S')$. We now show this is true for each iterations
of the outer \emph{for} loop of the algorithm.

  If the algorithm swaps $j(e)$ and $j$ then this swap reduces the total work load on $p_{i_e}$ from, say $\ell(p_{i_e})$ to $\ell(p_{i_e})-3+1=\ell(p_{i_e})-2$ and changes the total work load on $p_i$ from $\ell(p_i)$ to at most $\ell(p_i)-4+4=\ell(p_i)$.

  Otherwise, no dummy jobs or element jobs of work  4 is scheduled on $p_{i_e}$. The total work load on $p_{i_e}$ in this new schedule is at most $3 \cdot 1$, which is less than the cost of $j(e)$ alone in the old schedule.

  Since each job is considered once and all decisions only depend on some jobs' work, the algorithm run in polynomial time (all rebalancing can be done at the very end instead).
\qed\end{proof}

%

\begin{rem}
   The proof above does not assume that preemption is forbidden and so, it applies for both the preemptive  and non-preemptive cases.
   Moreover, since all the jobs have common release dates and deadlines the instance is \emph{agreeable}.
   Also, notice that the ratio $\frac{w_{\max}}{w_{\min}}$ equals 4 and so, the problem with constant ratio for which we gave an approximation
   algorithm in the previous section is already APX-Hard.
\end{rem}

\section{Single processor scheduling}\label{sec:single-proc}
In this section, we present a new LP relaxation for \textsc{Non preemptive minimum energy scheduling ($\alpha$)} on a single processor. 
\begin{thm}\label{thm:intgap}
The linear program  \textbf{LP1} has integrality gap at most $\lpgap$.
\end{thm}
As a corollary, we obtain the following theorem.

\begin{thm}\label{thm:existsingproc}
  There exists a polynomial-time algorithm which computes a $\singleprocgap$-approximation to the \textsc{Non preemptive minimum energy scheduling ($\alpha$)} problem on a single processor
\end{thm}
\begin{proof}
Since our proof of the integrality gap of \textbf{LP1} is algorithmic, it is straightforward to obtain the claimed ratio:

Given an instance of \textsc{Non preemptive minimum energy scheduling ($\alpha$)}, write the linear program \textbf{LP1} corresponding to this instance, solve it to obtain a fractional solution and then use Algorithm \ref{alg:singleproc} to obtain an integral solution of value at most $\lpgap$ times the value of the fractional solution and output it. Since the fractional solution obtained was optimal, by Lemma \ref{lem:discretization}, it has value at most $(1+\varepsilon)^{\alpha-1}$ times the energy consumed by the optimum schedule. Thus, the integral solution we output attains the claimed bound.
\qed \end{proof}

\subsection{Linear programming formulation}

To model the problem, we start from 0-1 variables $x_{I,j}$ indexed by a job $j$ and an execution interval $I$ which indicate whether $j$ is assigned to $I$. To bound the number of variables, we use the following result of Huang and Ott \cite{Huang_Ott} which allows us to restrict our attention to schedules where all execution intervals begin and end in some set  $T$ of time points such that $|T|$ is  polynomial in the input.

\begin{lem}[Discretization of time]\label{lem:discretization}\cite{Huang_Ott}
    Let $r_1,\ldots,r_{2n}$ be the release dates and deadlines of jobs.
    For each $1 \le i < 2n$, create $n^2 (1 + \frac{1}{\varepsilon}) -1$ equally-spaced ``landmarks'' in the interval $[r_i, r_{i+1}]$.
    Let $S$ be a solution of minimal cost such that for each job $j$ and each consecutive landmarks $t_i, t_{i+1}$, either
    job $j$ is executed during the whole interval $[t_i, t_{i+1}]$ or not at all.
    Then $E(S) \le (1 + \varepsilon)^{\alpha - 1} \text{OPT}$.
\end{lem}

Thus, we consider the set  ${\cal I}$ of all the  intervals with both endpoints in a landmark to be   the set of the allowed execution intervals.
Since $j$ must be scheduled somewhere, $\sum_{I} x_{I,j} = 1$. Since at any time $t$, at most one job is being processed, 
$\sum_j \sum_{I\ni t} x_{I,j}\leq 1$. Our LP, which we now state, contains an additional constraint (\ref{lp1-non-preempt}) capturing non-premption:
if some job $j$ is scheduled  during some interval $I$ or a subinterval thereof, then no other job can be scheduled during an interval that contains $I$. This holds for non-preemptive schedules but not necessarily for preemptive schedules, and in that sense this new constraint ``captures" non-preemption. The constraint is necessary to bound the integrality gap: without it, there exist instances and  fractional solutions that have much lower value than their integer counterpart (Lemma~\ref{lem:gap_lp1}, proved in below using  the instance described Figure~\ref{fig:exmpl_LP}.)   
%
\begin{alignat}{2}
    \textbf{\text{LP1:}} \quad \quad \text{minimize }   &E(\vec{x}) = \sum_{j \in J} \sum_{I \in \mathcal{I}} x_{I,j} \left(\frac{w_j}{|I|}\right)^{\alpha}  |I| &\\
    \text{subject to } & \sum_{I \in \mathcal{I}} x_{I,j} \ge 1 & \forall j \in J \label{lp1-job-assigned}\\
                       & \sum_{j \in J} \sum_{\substack{I \in \mathcal{I}\\ t \in I}} x_{I,j} \le 1 & \forall \text{ landmark } t \label{lp1-no-overlap}\\
                       & \sum_{\substack{I' \in \mathcal{I}\\ I' \cap I \neq \emptyset}} x_{I',j}  + \sum_{\substack{I'' \in \mathcal{I}\\ I \subseteq I'', j' \in J}} x_{I'',j'} \le 1& \forall I \in \mathcal{I}, \forall j \in J \label{lp1-non-preempt}\\
                       & x_{I,j} \ge 0 & \forall j \in J, \forall I \in \mathcal{I}(j) \label{lp1-non-negativity}\\
                       & x_{I,j} = 0 & \forall I \notin \mathcal{I}(j)
\end{alignat}

\begin{lem}\label{lem:gap_lp1}
Without constraint~(\ref{lp1-non-preempt}), \textbf{LP1} has integrality gap at least $\Omega (n^{\alpha-1})$.
\end{lem}
\begin{proof}
We construct an instance on which the integrality gap is at least $\Omega (n^{\alpha -1})$.
Let us define $n+1$ jobs. We create $n$ \emph{small} jobs and a \emph{big} job.
The $i^{\text{th}}$ small job has release date $i-1$ and deadline $i$ and processing requirement 1.
The big job has release date 0 and deadline $n$ and a processing requirement of $n$.
More details about this instance are given figure \ref{fig:exmpl_LP}.

An integral solution will have to process the big job between two consecutive small jobs and thus,
will have a cost of at least $(\frac{n+2}{2})^{\alpha} \cdot 2 = \Omega (n^{\alpha})$.

Now, consider the following assignment of the variables. For any job $j$, let $L_j^1$ and $L_j^2$ be respectively the
first and second halves of its life interval.
We set $x_{L_j^1,j} = x_{L_j^2,j} = 1/2$ and the other $x_{I,j}$ to 0. Hence constraint \ref{lp1-job-assigned} is satisfied.
Consider now a time $t$ which is not a release date or a deadline. At this time, we can process a small job and the big job.
For both of them, we are either on their first or second halves and so, the sum of the $x_{I,j}$ such that $I$ contains $t$ is
at most 1. It follows that constraint \ref{lp1-no-overlap} is also satisfied and that the assignment described is a solution for
the linear program without constraint \ref{lp1-non-preempt}.

By doing this, each job costs $2^{\alpha -1}$, hence the optimal fractional solution has value at most $n \cdot (2^{\alpha -1}) = \mathcal{O}(n)$.
Therefore, the integrality gap is at least $\Omega (n^{\alpha -1})$.
\qed\end{proof}

\begin{figure}
   \begin{center}
        \includegraphics[scale=0.9]{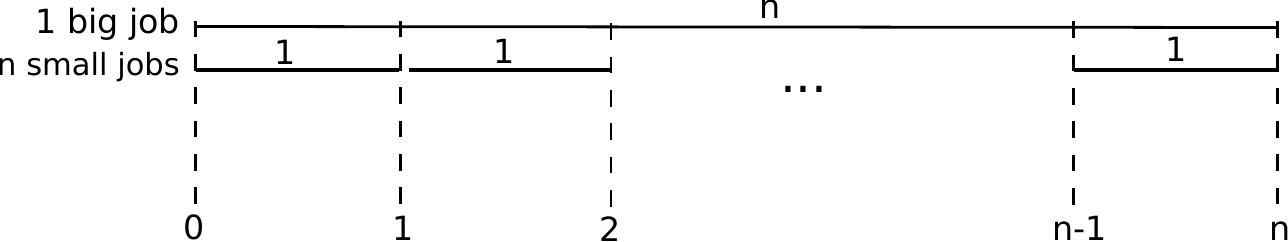}
   \end{center}
   \caption{An instance on which LP1 has an integrality gap of at least $\Omega (n^{\alpha -1})$.
   The instance contains a \emph{big} job and $n$ \emph{small} jobs. The figure shows the life intervals of the jobs
   and, above them, their work volume.}
   \label{fig:exmpl_LP}
\end{figure}

The remainder of this section is devoted to proving Theorem~\ref{thm:intgap}.

\subsection{Overview}

We show that any fractional solution can be transformed into an integral solution without increasing the value of the solution by too much, in three steps. 

We first divide the time into zones and transform the fractional solution so all (non-zero) fractional execution intervals are inside a zone. Then, each zone is divided into nested subzones and we further transform our fractional solution so that all fractional execution intervals are inside a subzone and the life interval of the corresponding job contains that subzone. Finally, we build a weighted bipartite graph from the transformed fractional solution whose edges represent the possible allocation of execution intervals to subzones. Similarly to~\cite{Shmoys_Tardos}, we find an integral (weighted) matching in this graph and translate this solution to an integral schedule.

We then show that the cost of the integral solution we built is at most $\lpgap$ times the cost of the original fractional solution.


\subsection{Building an integral solution from a fractional solution}

We now give the detailed description for our procedure to transform a fractional in three steps. The first step is derived from Antoniadis and Huang's algorithm \cite{Antoniadis_Huang}. 

\subsubsection{Splitting execution intervals on deadlines}

Our first transformation turns a fractional solution into a fractional solution where $x_{I,j}$ is 0 for any execution interval $I$ that any points in a set of deadlines we pick. The deadlines we pick are the deadlines of a good independent set.

\begin{lem}\label{lem:deadline_zones}
  Let \textbf{LP1} be the LP obtained from an instance of \textsc{Non preemptive minimum energy scheduling ($\alpha$)} and $\cal{J}$ be any good independent set for this instance.

  In polynomial time, we can transform any fractional solution $\vec{x}$ to \textbf{LP1} to a fractional solution $\vec{y}$ of value at most $2^{\alpha-1}E(\vec{x})$ where $y_{I,j}=0$ if $I$ crosses a deadlines of $\cal{J}$.
\end{lem}

To prove this lemma, we will simply ``shift'' some of the values of $\vec{x}$.
\begin{defn}
  By \emph{shifting} $y_{I,j}$ to $y_{I',j}$ for a fractional solution $\vec{y}$ to \textbf{LP1}, we mean to increase $y_{I',j}$ by $y_{I,j}$ and decrease $y_{I,j}$ to 0.
\end{defn}

\begin{proof}
  First note that $x_{I,j}=0$ if $I$ crosses \emph{two} deadlines of $\cal{J}$. Indeed, otherwise, $I$ contains the execution interval of some job $j'$ of the independent set and $\vec{x}$ doesn't satisfy constraint (\ref{lp1-no-overlap}).

  We build $\vec{y}$ iteratively, starting with $\vec{y}=\vec{x}$.

  For each $x_{I,j}>0$ and $I=[s,e]$ crosses a deadline $d$ in $\cal{J}$ (so $s<d<e$), $I'$ the larger of the two intervals $[s,d]$, $[d,e]$ has size at least half the size of $I$. We shift $y_{I,j}$ to $y_{I',j}$. 

  This process shrink the size of each execution interval in the fractional solution by a factor of at most 2 and thus $E(\vec{y})$ is at most $2^{\alpha-1}E(\vec{x})$.

  Constraint (\ref{lp1-job-assigned}) is still satisfied as shifting $x_{I,j}$ from $y_{I,j}$ to $y_{I',j}$ preserves the sum of probabilities over $j$ (and does not affect the constraint for any other job $j' \ne j$).
  Constraint (\ref{lp1-no-overlap}) is still satisfied as $I'$ is contained in $I$ (so fewer (fractional) execution intervals cross each deadline in $\cal{J}$).
  Again, since $I'$ is contained in $I$, we see that constraint (\ref{lp1-non-preempt}) is satisfied since an inequality in constraint (\ref{lp1-non-preempt}) containing the term $x_{I',j}$ contains the term $x_{I,j}$.
\qed\end{proof}

\subsubsection{Further splits}

We now proceed to our second transformation and show that we can further split the execution intervals of a fractional solution. Now that all execution intervals
(in the support of $\vec{x}$) lie between two consecutive deadlines (which we now call a ``zone''), we can further partition each zone so the first half
is dedicated to jobs whose life interval ends in that zone and the second half is dedicated to the others (namely, jobs whose life interval starts
in that zone or jobs whose life interval contains the zone).

\begin{lem}\label{lemma:subzones}
  Let \textbf{LP1} be obtained from an instance of \textsc{Non preemptive minimum energy scheduling ($\alpha$)} and $\cal{J}$ be any good independent set for this instance.

  In polynomial time, we can transform any fractional solution $\vec{y}$ where $y_{I,j}=0$ if $I$ crosses a deadlines of $\cal{J}$ to \textbf{LP1} to a fractional solution $\vec{z}$ of value at most $2^{\alpha-1}$ times the value of $\vec{y}$  where $z_{I,j}>0$ implies
  \begin{enumerate}
  \item
    $I \subseteq [d_s, d_s + \frac{1}{2^{k}}(d_e-d_s) ] \subseteq L_j$ for some consecutive deadlines $d_s, d_e$  of $\cal{J}$ and $k \ge 1$, or
  \item
    $I \subseteq [d_e - \frac{1}{2^{k}}(d_e-d_s), d_e ] \subseteq L_j$ for some consecutive deadlines $d_s, d_e$  of $\cal{J}$ and $k \ge 1$.
  \end{enumerate}
\end{lem}

We let ${\cal Z}$ consists of all intervals of the form $[d_s, d_s + \frac{1}{2^{k}}(d_e-d_s) ]$ and $[d_e - \frac{1}{2^{k}}(d_e-d_s), d_e]$ 
for consecutive deadlines $d_s, d_e$ of $\cal{J}$.

Though they do not partition the timeline, we still refer to ${\cal Z}$ as \emph{subzones}. We now prove the following refinement of the above lemma where we make ${\cal Z}$ explicit.

\begin{proof}
  First set $\vec{z}$ to be $\vec{y}$.

  Let $[d_s,d_e]$ be a zone defined by two consecutive deadlines. No life interval $L_j$ is contained in $[d_s,d_e]$ (or $\cal{J}$ is not a good independent set as we could add $L_j$ and still obtain an independent set). Thus, we can partition jobs whose life interval intersects $[d_s,d_e]$ into jobs $E$ whose life interval starts (at or) before $d_s$ and jobs $S$ whose life interval end (at or) after $d_e$ (putting jobs that can go into both into either set).

  If $z_{I,j}>0$ and $I$ intersects $[d_s,d_e]$ then $I$ is a subset of $L_j$ which also intersects $[d_s,d_e]$ so $j$ is in $S$ or $E$.

  We now describe how to shift $z_{I,j}$ with $j$ in $E$ (the shift for jobs in $S$ is symmetric). We simulaneously shift all $z_{I,j}>0$ with $I=[s,e]$. Since $I$ does not cross $d_s$ or $d_e$, it ends in $[d_s,d_e]$ and $e$ is in $[d_s + \frac{1}{2^k}, d_s + \frac{1}{2^{k-1}}(d_e-d_s)]$ for some $k$. We shift $z_{I,j}$ to $z_{I',j}$ where $I=[s-\frac{1}{2}(s-d_s) , e-\frac{1}{2}(e-d_s)$.

  Shifts for $z_{I,j}$ with $j$ in $S$ are defined symmetrically (by reversing the timeline).

  Since each execution interval of $\vec{y}$ is shifted to an execution interval exactly half its original size, by Lemma \ref{lemma:half}, $E(\vec{z})$ is at most $2^{\alpha-1}E(\vec{y})$.
  Since we only shifted intervals, constraint (\ref{lp1-job-assigned}) remains satisfied. Constraint (\ref{lp1-no-overlap}) is satisfied as we simply compressed the entire region $[d_s, d_e]$ to $[d_s, d_e-\frac{1}{2}(d_e-d_s)]$ for jobs in $S$ and the entire region $[d_s, d_e]$ to $[d_s +\frac{1}{2}(d_e-d_s), d_e]$ for jobs in $E$. Finally, constraint (\ref{lp1-non-preempt}) is satisfied by $\vec{z}$ since again these constraint were satisfied by $\vec{y}$ and we only compressed some block of execution intervals into disjoint regions of the timeline.
\qed\end{proof}

\subsubsection{Building a weighted bipartite matching}

As a result of Lemma \ref{lemma:subzones}, for each $z_{I,j}>0$, $I$ is contained in some subzone $Z$ and furthermore, the life interval of $j$ contains $Z$ so we can freely shift $I$ to another interval (of the same length) inside $Z$.
Thus, we will only remember the length of the fractional execution intervals and the subzone $Z$ in which they belong. I.e., we think of $\vec{z}$ as a fractional assignment of lengths $\ell_i$ for each job to $Z$.

\begin{lem}\label{lemma:earliest_deadline_first}
  If for each subzone $Z$, the lengths $\ell(e)$ assigned to $Z$ and all subzones included in $Z$ is at most $|Z|$ then there is a feasible schedule where each job is given their assigned length in $Z$.
\end{lem}
\begin{proof}
We greedily assign $j$ an interval of length $\ell(e)$ to the leftmost possible empty spot if $Z$ is of the form $[d_s, d_s + \frac{1}{2^{k}}(d_e-d_s) ]$. In fact, at each intermediate step, all jobs currently assigned to $Z$ take up a contiguous region starting from the left endpoint $d_s$ of $Z$. If $Z$ is of the form $[d_e - \frac{1}{2^{k}}(d_e-d_s), d_e ]$, we assign $j$ the rightmost possible interval of length $\ell(e)$.
We see by induction on the number of subzones contained in $Z$ that every job is processed in $Z$.
\end{proof}


We now desire an integral assignment of lengths to each $Z$ where the total of all lengths assigned to $Z$ does not exceed $|Z|$. Note that this constraint is satisfied by the fractional solution derived from $\vec{z}$ (as $\vec{z}$ satisfies constraint (\ref{lp1-no-overlap})).

To obtain such an integral assignment from our fractional assignment derived from $\vec{z}$, we build a weighted bipartite graph $G(\vec{z})$ where an assignments correspond to matchings and the weight of a matching correspond to the energy cost (of the matching interpreted as a schedule). We will then obtain an integral matching from the derived fractional matching (whose weight is exactly $E(\vec{z})$.

We now describe $G(\vec{z})$ with bipartition $(A,B)$ and weight $w(e)$ for each edge $e \in E(G)$. We also keep a \emph{length} $\ell(e)$ for each edges which will be used in the very last step of our proof (but in no way affects the weighted bipartite matching we look for).

\begin{itemize}
\item
  $A$ contains one vertex for each job. I.e., $A = \{a_j | j \in J\}$
\item
  $B$ consists of vertices for subzones. However, $B$ may contain more than one vertex for each subzone $Z$. In fact, it contains the ceiling of the sum of fractional value of all lengths assigned to $Z$. I.e.,
  \[
  B = \left\{b_{Z,i} | Z \in {\cal Z}, i \in 1,\ldots, \left\lceil \sum_{j \in J} \sum_{I \subseteq Z } z_{I,j} \right\rceil \right\}
  \]

\item
  The edges are constructed as followed. Start with all edges $a_jb_{Z,i}$ for all $i$ if $z_{I,j}>0$ for some $I \subseteq Z$. We now delete some edges to obtain the edges of $G(\vec{z})$ and assign weights and length of the remaining edges.

  Sort the lengths assigned to $Z$ by $\vec{z}$ in decreasing order of length. For each such length $\ell_k$ for job $j$ of fractional value $z_{I,j}$, set $w(a_jb_{Z,i})$ to $\frac{w_j^\alpha}{\ell_k^{\alpha-1}}$ where $i$ is the ceiling of the partial sum of all jobs previously considered for $Z$ (i.e., $i=\lceil \sum_{q=1}^{k-1} \ell_{q} \rceil$). Set $\ell(a_jb_{Z,i})$ to $\ell_k$. Also set $w(a_jb_{Z,i+1})$ to $\frac{w_j^\alpha}{\ell_k^{\alpha-1}}$ and $\ell(a_jb_{Z,i+1})$ to $\ell_k$ if adding $j$ to the ceiling of the partial sum increases it by 1. Delete all other edges of the form $a_jb_{Z,t}$.
\end{itemize}

$\vec{z}$ naturally gives the following fractional matching $M(\vec{z})$ of $G(\vec{z})$ with total weight $E(\vec{z})$: we pick each edge with weight exactly $z_k$ (or $z_k$ split into two as follows if adding $z_k$ increased the ceiling of the partial sum by 1. Whatever we need to add to the partial sum to make it an integer is the fraction we choose of the first edge, and the rest of $z_k$ for the second edge).

To complete the description of our final transformation, we apply the following two technical lemmas.

\begin{lem}\label{lemma:integer-matching}\cite{Lovasz_Plummer_matching,Shmoys_Tardos}
  In a weighted bipartite graph, there exists an (integral) matching of same weight as any fractional matching\footnote{in a fractional matching edges can be selected with a fractional value as long as the total value of edges incident to any vertex is at most 1}.
\end{lem}

\begin{lem}\label{lemma:matching-to-schedule}
  Let $G(\vec{z})$ be the bipartite graph built from a transformed fractional solution $\vec{z}$.
  For any matching $M$ saturating $A$ of $G(\vec{z})$, we can obtain a schedule whose energy consumption is at most $3^{\alpha-1}$ times the weight of $M$. 
\end{lem}
\begin{proof}

\textbf{Schedule construction}.
Since $M(\vec{z})$ has weight $E(\vec{z})$ and saturates all of $A$, by Lemma \ref{lemma:integer-matching} there exists an integer matching $M'$ of the same weight that saturates $A$.

We build a schedule $S$ from $M'$ as follows. Consider subzones in order of containment starting with subzones containing no other subzones. For each edge $e=(a_j, b_{Z,i}) \in M'$ give $j$ an execution interval of length $\ell(e)/3$ in $Z$. By Lemma \ref{lemma:matching-to-schedule}, we only need to verify that for all subzones $Z$, the sum of lengths for all subzones contained in $Z$ does not exceed $|Z|$.

\textbf{Feasibility}.
We now check that for each subzone $Z$, the sum of lengths for all subzones contained in $Z$ does not exceed $|Z|$.

\begin{figure}[ht!]
   \begin{center}
        \includegraphics{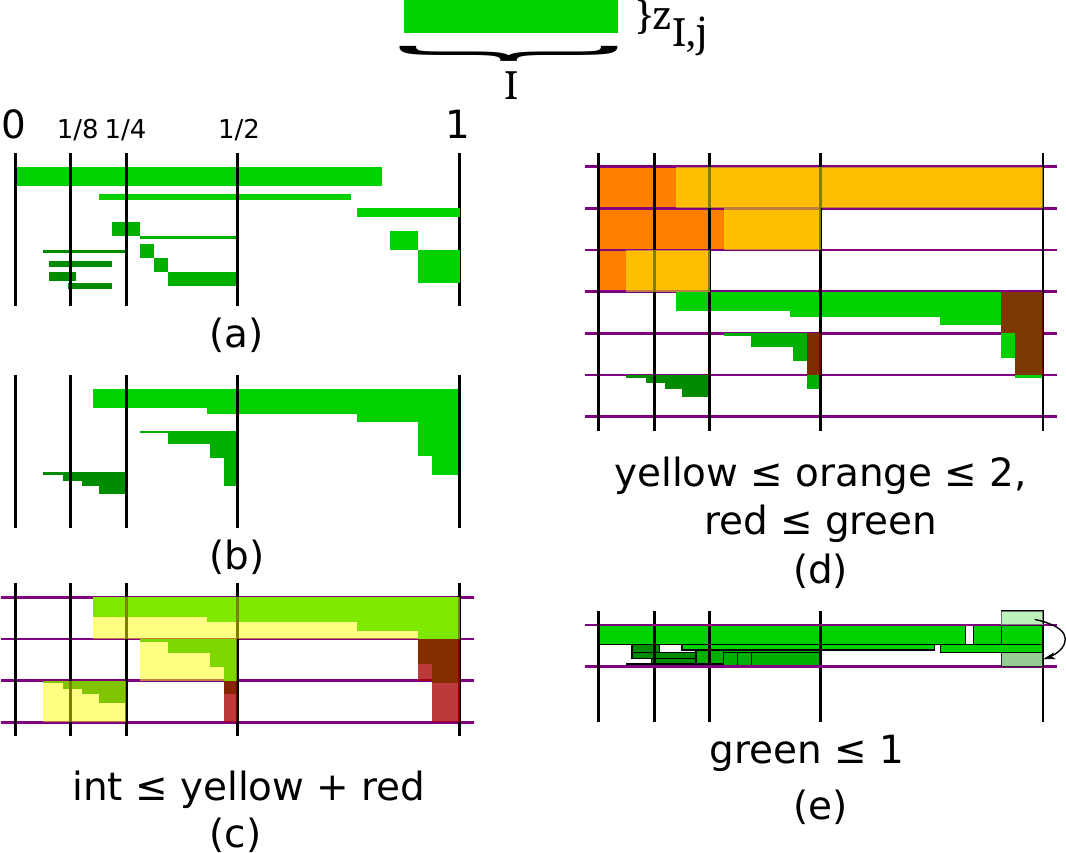}
   \end{center}
   \caption{An illustration of the feasibility proof. The fractional solution is represented in green with $z_{I,j}$ corresponding to the height of a rectangle. (b) Recall that $G(\vec{z})$ is built by first ordering the execution intervals by length. (c) In the worst case, the matching picked the longest edge incident to each vertex.}
   \label{fig:LP_matching}
\end{figure}

Let $v(Z)$ be the sum of all lengths assigned to $Z$ by the fractional solution, weighted by $z_{I,j}$ (i.e., $v(Z) = \sum_{j, I, I \subseteq Z \subseteq L_j} z_{I,j} |I|$). Since $\vec{z}$ satisfies constraint (\ref{lp1-no-overlap}), we have
\[
\sum_{Z' \subseteq Z} v(Z') \le |Z|
\]
for all $Z$.


Let $mi(Z,i)$ be the minimum length of all edges incident to $b_{Z,i}$ and $ma(Z,i)$ be the maximum length of all edges incident to $b_{Z,i}$. Since we considered lengths in non-increasing order of their $z$ values, $mi(Z,i) \ge ma(Z,i+1)$ for all $i$ and $Z$.

Let $n(Z)$ is the number of copies of vertex $b_{Z,i}$ and $\ell(Z)$ be the sum of lengths of all edges of $M$ incident to $\{b_{Z,i}\}_{i=1}^{n(Z)}$.
For all $Z$,
\[
\ell(Z) \le \sum_{i=1}^{n(Z)} ma(Z,i) \le ma(Z,1) + \sum_{i}^{n(Z)-1} mi(Z,i) \le |Z| + v(Z).
\]
and therefore the sum of lengths assigned to $Z$ in our schedule is at most $\frac{1}{3}(|Z| + v(Z))$.

Now for each $Z$, the sum of all lengths assigned to $Z$ and all subzones included in $Z$ in our schedule is at most
\[
\frac{1}{3}\left(\sum_{Z' \subseteq Z} |Z'|+ \sum_{Z' \subseteq Z} v(Z')\right) \le \frac{1}{3}\left(\left(\sum_{q=0}^\infty \frac{|Z|}{2^i}\right) + |Z|\right) \le \frac{1}{3}(3|Z|)
\]
and so our assigned lengths are feasible.
\qed\end{proof}

\begin{proof}(of Theorem \ref{thm:intgap})
  To prove the integrality gap of \textbf{LP1}, we simply need to apply each lemma in this section in turn.

  Given an optimal fractional solution $\vec{x}$ to \textbf{LP1}, find a good independant set $\cal{J}$ (to the instance which generate the LP) and apply Lemma \ref{lem:deadline_zones} to $\vec{x}$ and $\cal{J}$ to obtain a fractional solution $\vec{y}$ of value at most $2^{\alpha-1}E(\vec{x})$ where no execution interval crosses a deadline in $\cal{J}$.

  Then apply Lemma \ref{lemma:subzones} to $\vec{y}$ to obtain $\vec{z}$ of value at most $2^{\alpha-1}E(\vec{y}) \le 4^{\alpha-1}E(\vec{x})$ where for any non-zero $z_{I,j}$, $I$ is contained in a ``subzone'' of the form $[d_s, d_s + \frac{1}{2^{k}}(d_e-d_s)]$ or $[d_e - \frac{1}{2^{k}}(d_e-d_s), d_e ]$ for some consecutive deadlines $d_s, d_e$ and $k \ge 1$ and furthermore $L_j$ contains this subzone.

  Now build $G(\vec{z})$ and interpret $\vec{z}$ as a fractional matching in $G(\vec{z})$. By Lemma \ref{lemma:integer-matching}, $G(\vec{z})$ has a matching $M$ of same weight as this fractional matching and by Lemma \ref{lemma:matching-to-schedule}, we can build a schedule from $M$ whose energy consumption is at most $3^{\alpha-1}E(\vec{z}) \le 12^{\alpha-1}E(\vec{x})$.

  Such a schedule is of course a solution to \textbf{LP1} of same value, thus completing the proof.
\qed\end{proof}




\subsubsection{Algorithm summary}

We can summarize our algorithm from transforming any fractional solution $\vec{x}$ to an integral solution.

\begin{alg}\label{alg:singleproc}
~
\begin{enumerate}
\item Apply the transformation of Lemma \ref{lem:deadline_zones} and then Lemma \ref{lemma:subzones} to the fractional solution to obtain a new fractional solution $\vec{z}$.
\item Construct the weighted bipartite graph $G(\vec{z}) = (A,B)$.
\item Find a minimum weight matching $M$ that matches every node in $A$.
\item For each edge $e = (a_j, b_{Z, i}) \in M$, schedule job $j$ in the subzone $Z$ with an interval of length $\frac{\ell(e)}{3}$. Use an earliest deadline first schedule for all jobs in $Z$ if $Z$ is in the first half of a zone and an latest release date first schedule if $Z$ is in the second half of a zone.
\end{enumerate}

\end{alg}

\section{Open problems}

The APX-Hardness result and the algorithm for the heterogeneous multiprocessor setting togethers with
the linear program for the non-preemptive single processor case open several natural research directions. We summarize a few of them below.

\begin{itemize}
\item Is our analysis of the integrality gap of the linear program for $1|r_j,d_j|E$ tight? Is there an instance for which there is a gap, any gap at all 
  (beyond the gap created from discretization of time)?
\item Since we gave both an APX-Hardness proof and a constant factor approximation algorithm for the non-preemptive speed scaling problem with heterogeneous 
  constant work volumes, what is the best possible constant we can hope for?
\item The linear program we gave for $1|r_j,d_j|E$ has a natural extension for the multiprocessor setting (create a variable $x_{I,j,p}$ for any interval $I$,
  job $j$ and processor $p$). This extension has an arbitrarly large integrality gap (a refinement of the example of figure \ref{fig:exmpl_LP} leads to a gap of
  at least $\Omega (\sqrt(n)^{\alpha -1})$
  for $2|r_j,d_j|E$). Are there stronger linear constraints that we can add in order to obtain a constant integrality gap on two processors? One can show that we can force
  an independent set to be executed on one of the two processors, does this leads to a constant integrality gap? 
\item We showed that $R|r_j,d_j, w_{i,j},\text{pmtn, no-mig}|E$ is APX-Hard, is the non-heterogeneous non-preemptive version of the problem, namely $P|r_j,d_j,w_j|E$, also APX-Hard 
  (are work volume heterogeneity and non-preemption \emph{equally} hard?)? 
\end{itemize}

\bibliographystyle{plain}
\bibliography{non-preemptive_results}

\end{document}